\documentclass[11pt]{article}

\usepackage{amsxtra,amssymb,amsthm,amsmath,latexsym}
\usepackage{graphicx}
\usepackage{epsfig}
\usepackage{afterpage}
\newtheorem{thm}{Theorem}[section]

\newtheorem{lem}[thm]{Lemma}
 
 \newcommand{\lemref}[1]{Lemma~\ref{#1}}

\newcommand{\R}{{\mathbb R}}

\newcommand{\dl}{{\delta}}

\newcommand{\bee}{\begin{equation*}}
\newcommand{\eee}{\end{equation*}}
\newcommand{\be}{\begin{equation}}
\newcommand{\ee}{\end{equation}}
\newcommand{\pn}{\par\noindent}
\title{A singular integral equation for electromagnetic wave scattering}
\author{A G Ramm\\
\small Department of Mathematics\\[-0.8ex]
\small Kansas State University, Manhattan, KS 66506-2602, USA\\[-0.8ex]
\small \texttt{ramm@math.ksu.edu}\\
}

\begin{document}
Intern.Journ. Pure and Appl. Math., 55, N4, (2009), 7-11.
\date{}
% typeset front matter
\maketitle

\begin{abstract}A 3D singular integral equation is derived for
electromagnetic wave scattering by bodies of arbitrary shape.
Its numerical solution by a projection method is outlined.
\end{abstract}
\pn{
 }

\section{Introduction}
Consider the following scattering problem. An incident
electromagnetic field $(E_0,H_0)$ is scattered by a bounded region
$D$, filled with a material with parameters
$(\epsilon,\sigma,\mu_0)$. The exterior region $D'$ is a homogeneous
region with parameters $(\epsilon_0,\sigma=0,\mu_0)$. Consider for
simplicity the case when $\epsilon=const$ in $D$, $\sigma=const\geq
0$ in $D.$ Let $\epsilon'=\epsilon+\frac{i\sigma}{\omega}$. The
governing equations in $\R^3$ are \be\label{e1} \nabla\times
E=i\omega \mu_0H,\quad \nabla\times H=-i\omega \epsilon' E. \ee At
the boundary $S$ of $D$ one has \be\label{e2} [N,E^+]=[N,E^-], \ee
and \be\label{e3} N\cdot \epsilon' E^+=N\cdot \epsilon_0 E^-, \ee
where $N$ is the unit normal to $S$, pointing into $D'$,
$E^+(E^-)$ is the limiting value of $E$ on $S$ from inside (outside)
$S$, $[N,E]$ is the cross product, and $E\cdot N$ is the dot product
of two vectors.

Let \be\label{e4} k^2=\omega^2\epsilon_0\mu_0,\quad
K^2=\omega^2\epsilon'\mu_0,\quad \mathcal{K}^2=\left\{
                                                 \begin{array}{ll}
                                                   k^2, & \hbox{in $D'$,} \\
                                                   K^2, & \hbox{in $D$.}
                                                 \end{array}
                                               \right.
 \ee Equations \eqref{e1} imply
\be\label{e5} \nabla\times\nabla\times E-\mathcal{K}^2E=0,\quad
H=\frac{\nabla \times E}{i\omega \mu_0}\quad in \ \R^3. \ee
Therefore, it is sufficient to find $E$ satisfying the first
equation \eqref{e5}, boundary conditions \eqref{e2}, \eqref{e3}, and
the radiation condition 
\be\label{e6} E=E_0+V;\quad
V_r-ikV=o\left(\frac{1}{r}\right),\qquad r:=|x|\to \infty.\ee

Equation \eqref{e5} for $E$ can be written as \be\label{e7}
LE:=\nabla\times\nabla\times E-k^2E=pE;\quad
p:=p(x)=\mathcal{K}^2-k^2=\left\{
                            \begin{array}{ll}
                              0, & \hbox{in $D'$,} \\
                              K^2-k^2, & \hbox{in $D$.}
                            \end{array}
                          \right.
 \ee The incident field $E_0$ solves equation \eqref{e7} with $p=0.$

Let $\dl(x)$ denote the delta-function and $\dl_{ij}:=\left\{
                                                        \begin{array}{ll}
                                                          1, & \hbox{$i=j$,} \\
                                                          0, & \hbox{$i\neq j$.}
                                                        \end{array}
                                                      \right.$

Let $G=G_{ij}(x)$ solve the problem: \be\label{e8}
LG=\dl(x)\dl_{ij},\quad G_r-ikG=o\left(\frac{1}{r}\right),\quad r\to
\infty. \ee 
Then the solution to \eqref{e7}-\eqref{e6} solves the
integral equation 
\be\label{e9} E=E_0+\int_{\R^3}G(x-y)p(y)E(y)dy.
\ee The kernel $G(x)=G(|x|)$ is symmetric, $G_{ij}=G_{ji}$, see
formula \eqref{e15}. Let us prove
\begin{lem}\label{lem1}
There is at most one solution to \eqref{e9} satisfying \eqref{e2}
and \eqref{e3}.
\end{lem}
\begin{proof}
If there are two solutions then their difference $E$ solves the
homogeneous equation \eqref{e9} and satisfies \eqref{e2} and
\eqref{e3}. Thus, $E$ solves \eqref{e7}, \eqref{e6}, \eqref{e2} and
\eqref{e3}. Therefore, $E$ and $H=\frac{\nabla\times E}{i\omega
\mu_0}$ solve equations \eqref{e1} and satisfy condition \eqref{e2},
\eqref{e3} and \eqref{e6}. It is known (see e.g., \cite{M}) that
this implies $E=H=0$. \\
\lemref{lem1} is proved.
\end{proof}
\begin{lem}\label{lem2}
If $E$ solves equation \eqref{e9}, then it satisfies \eqref{e6},
\eqref{e2}, \eqref{e3} and \eqref{e7}. Therefore, \eqref{e9} has at
most one solution.
\end{lem}
\begin{proof}
Applying operator $L$ to \eqref{e9} one obtains equation \eqref{e7}.
The integral in \eqref{e9} is the term $V$ in \eqref{e6}. It
satisfies the radiation condition because $G$ does. Equation
\eqref{e7} is equivalent to \eqref{e5}. Equation \eqref{e5} together
with the formula $H=\frac{\nabla \times E}{i\omega \mu_0}$ yield
both equations \eqref{e1}. Conditions \eqref{e2} and \eqref{e3} are
consequences of equations \eqref{e1}. Therefore, every solution to
\eqref{e9} is in one-to one correspondence with the solution to
equations \eqref{e1}. This correspondence is given by the formulas
$E=E$, $H=\frac{\nabla \times E}{i\omega \mu_0}$. By \lemref{lem1}
equation \eqref{e9} has at most one solution satisfying \eqref{e2}
and \eqref{e3}. We have proved that every solution to \eqref{e9}
satisfies \eqref{e2} and \eqref{e3}. Therefore, \eqref{e9} has at
most one solution. \\
\lemref{lem2} is proved.
\end{proof}
\begin{lem}\label{lem3}
Equation \eqref{e9} has a unique solution.
\end{lem}
\begin{proof}
Uniqueness of the solution to \eqref{e9} is proved in \lemref{lem2}.
Existence of it follows from the existence of the solution to the
scattering problem and the fact, established in the proof of
\lemref{lem2}, that a solution to \eqref{e9} solves equation
\eqref{e5} and satisfies the radiation condition \eqref{e6} and
conditions \eqref{e2}, \eqref{e3}. \\
\lemref{lem3} is proved.
\end{proof}

From lemmas 1-3 the following result follows
\begin{thm}
Equation \eqref{e9} has a unique solution $E$. This solution $E$
generates the solution to the scattering problem by the formula
$E=E$, $H=\frac{\nabla \times E}{i\omega \mu_0}.$
\end{thm}

In Section 2 we construct the Green's function $G.$
\section{Construction of $G$}
Let us look for $G$ of the form \be\label{e10}
G(x)=\int_{\R^3}e^{i\xi\cdot x}\tilde{G}(\xi)d\xi,\quad
\tilde{G}(\xi)=\frac{1}{(2\pi)^3}\int_{\R^3}e^{-i\xi\cdot x}G(x)dx.
\ee Take the Fourier transform of \eqref{e8} and get 
\be\label{e11}
-[\xi,[\xi,\tilde{G}]]-k^2\tilde{G}=\frac{1}{(2\pi)^3}I,\quad
I_{ij}=\dl_{ij}, \ee where $[a,b]$ is the cross product of two
vectors, and $a\cdot b$ is their dot product. This implies
\be\label{e12}
-\xi\xi\cdot\tilde{G}+(\xi^2-k^2)\tilde{G}=\frac{1}{(2\pi)^3}I. \ee
From \eqref{e12} one finds \be \xi \cdot
\tilde{G}=-\frac{\xi}{(2\pi)^3k^2}. \ee Thus, \be\label{e14}
\tilde{G}_{ij}=\frac{\dl_{ij}}{(2\pi)^3(\xi^2-k^2)}-\frac{\xi_i\xi_j}{(2\pi)^3k^2(\xi^2-k^2)}.
\ee 
Taking the inverse Fourier transform and using the radiation
condition \eqref{e6}, one gets \be\label{e15}
G_{ij}(x)=g(x)\dl_{ij}+\frac{1}{k^2}\partial_{ij}g(x);\quad
g=\frac{e^{ik|x|}}{4\pi|x|},\quad
\partial_i:=\frac{\partial}{\partial x_i}. \ee
From \eqref{e15} and \eqref{e9} one gets:
\be\label{e16}\begin{split}
E_i(x)&=E_{0i}(x)+(K^2-k^2)\int_Dg(x,y)E_i(y)dy\\
&+\frac{K^2-k^2 }{k^2}\frac{\partial}{\partial
x_i}\int_D\frac{\partial g(x,y)}{\partial x_j}E_j(y)dy,\quad 1\leq
i\leq 3,\end{split}\ee where summation over the repeated indices is
understood. Equation \eqref{e16} is a vector singular integral
equation. The operator 
$$TE=(K^2-k^2)\int_D g(x,y)E(y)dy $$ is compact
in $L^2(D)$. Let 
\be\label{e17} \gamma:=\frac{K^2-k^2}{k^2},\qquad
QE=\nabla \int_D\nabla_xg(x,y)\cdot E(y)dy. \ee 
Then equation \eqref{e16} can
be written as \be\label{e18} E=E_0+TE+\gamma QE. \ee This is
equation \eqref{e9}.

Numerically one can solve equation \eqref{e16} (or \eqref{e18}) by a
projection method. For example, let $\{\phi_j(x)\}$ be a basis of
$L^2(D)$ and $\phi_j\in H^{1}_0(D)$, where $H^{1}_0(D)$ is the
closure of $C_0^\infty(D)$ functions in the norm of the Sobolev
space $H^1(D)$. Multiply equation \eqref{e16} by $\overline{\phi}_m$
(the bar stands for the complex conjugate), integrate over $D$ and
then the third term by parts, to get: 
\be\label{e19}\begin{split}
&E_{im}=E_{0im}+(K^2-k^2)\int_D\int_D dx
\overline{\phi}_m(x)g(x,y)\sum_{m'=1}^ME_{im'}\phi_{m'}(y)\\
&-\gamma\int_Ddx\frac{\partial\overline{\phi}_m(x)}{\partial
x_i}\int_D\frac{\partial g(x,y)}{\partial
x_j}\sum_{m'=1}^ME_{jm'}\phi_{m'}(y)dy,\quad 1\leq m\leq M,\ 1\leq
i\leq 3.
\end{split}\ee
This is a linear algebraic system for finding the coefficients:
\be\label{e20}
E_{im}^{(M)}:=E_{im}:=\int_DE_i(x)\overline{\phi}_m(x)dx .\ee The
number $M$ determines the accuracy of the appproximate solution
$E(x)$. One has \be\label{e21} \lim_{M\to
\infty}\|E^{(M)}(x)-E(x)\|_{L^2(D)}=0. \ee

\end{document}